\documentclass{article}
\usepackage[utf8]{inputenc}
\usepackage[T1]{fontenc}
\usepackage{amsmath}
\usepackage{amssymb}
\usepackage{amsthm}
\usepackage{microtype}
\usepackage{cite}
\usepackage{enumerate}
\usepackage{graphicx}
\usepackage{xspace}
\usepackage{hyperref}
\usepackage{xcolor}
\usepackage{tikz}
\usetikzlibrary{arrows}
\usepackage{etoolbox}

\newtoggle{long}
\toggletrue{long}
\newcommand{\LV}[1]{\iftoggle{long}{#1}{}}

\newcommand{\SV}[1]{\iftoggle{long}{}{#1}}

\DeclareMathOperator{\Perf}{Perf}
\newcommand{\N}{\mathbb{N}}
\newcommand{\AlgSimple}{\ensuremath{\textsc{Simple}}}
\newcommand{\AlgKarl}{\ensuremath{\textsc{Linear}}}
\newcommand{\AlgJakub}{\ensuremath{\textsc{Binary}}}

\newtheorem{theorem}{Theorem}{\bfseries}{\itshape}
\newtheorem{lemma}{Lemma}{\bfseries}{\itshape}

\graphicspath{{./pictures/}}

\title{Online Checkpointing with Improved Worst-Case Guarantees}
\author{Karl Bringmann\thanks{Karl Bringmann is a recipient of the \emph{Google Europe Fellowship in Randomized Algorithms}, and this research is supported in part by this Google Fellowship.}
\and Benjamin Doerr\and Adrian Neumann\and Jakub Sliacan}

\begin{document}
\maketitle
\begin{abstract} 
  In the online checkpointing problem, the task is to continuously maintain a set of $k$ checkpoints that allow to rewind an ongoing computation faster than by a full restart. The only operation allowed is to replace an old checkpoint by the current state. Our aim are checkpoint placement strategies that minimize rewinding cost, i.e., such that at all times $T$ when requested to rewind to some time $t \le T$ the number of computation steps that need to be redone to get to $t$ from a checkpoint before $t$ is as small as possible. In particular, we want that the closest checkpoint earlier than $t$ is not further away from $t$ than $q_k$ times the ideal distance $T / (k+1)$, where $q_k$ is a small constant.
 
  Improving over earlier work showing $1 + 1/k \le q_k \le 2$, we show that $q_k$ can be chosen asymptotically less than $2$. We present algorithms with asymptotic discrepancy $q_k \le 1.59 + o(1)$ valid for all $k$ and $q_k \le \ln(4) + o(1) \le 1.39 + o(1)$ valid for $k$ being a power of two. Experiments indicate the uniform bound $p_k \le 1.7$ for all $k$. For small $k$, we show how to use a linear programming approach to compute good checkpointing algorithms. This gives discrepancies of less than $1.55$ for all $k < 60$. 
  
  We prove the first lower bound that is asymptotically more than one, namely $q_k \ge 1.30 - o(1)$. We also show that optimal algorithms (yielding the infimum discrepancy) exist for all~$k$. 
\end{abstract}

\section{Introduction}

Checkpointing means storing \LV{selected }intermediate states of a long sequence of computations. This allows reverting the system to a\LV{n arbitrary} previous state much faster, since only the computations from the preceding checkpoint have to be redone. Checkpointing is one of the fundamental techniques in computer science. Classic results date back to the seventies~\cite{ChandyR72}, more recent topics are checkpointing in distributed~\cite{ElnozahyAWJ02}, sensor network~\cite{Osterlindetal09}, or cloud~\cite{YiKA10} architectures.

Checkpointing usually involves a \LV{careful }trade-off between the speed-up of reversions to previous states and the costs incurred by setting checkpoints (time, memory). Much of the classic literature (see~\cite{Gelenbe79} and the references therein) studies checkpointing with the focus of \LV{gaining }fault tolerance against immediately detectable faults. Consequently, only reversions to the most recent checkpoint are needed. However, setting a checkpoint can be \SV{expensive}\LV{highly time consuming}, because the whole system state has to be copied to secondary memory. In such scenarios, the central question is how often to set a checkpoint such that the expected time spent on setting checkpoints and redoing computations from the last checkpoint is minimized (under a stochastic failure model and further, possibly time-dependent~\cite{TouegB84}, assumptions on the cost of setting a checkpoint).

In this work, we will regard a checkpointing problem of a different nature. If not fault-tolerance of the system is the aim of checkpointing, then often the checkpoints can be kept in main memory. Applications of this type arise in data compression~\cite{BernGRS94} and numerics~\cite{HeuvelineW06,StummW10}. In such scenarios, the cost of setting a checkpoint is small compared to the cost of the regular computation. Consequently, the memory used by the stored checkpoints is the bottleneck.

The first to provide a\LV{n abstract} framework independent of a particular application\LV{ in mind} were Ahlroth, Pottonen and Schumacher~\cite{AhlrothPS11}. They do not make assumptions on which reversion \LV{to previous states }will be requested, but simply investigate how checkpoints can be set in an online fashion such that at all times their distribution is balanced over the \LV{total }computation history.

They assume that the system is able to store up to $k$ checkpoints (plus a free checkpoint at time $0$). At any point in time, a \LV{previous }checkpoint may be discarded and replaced by the current \LV{system }state as new checkpoint. Costs incurred by such a change are ignored. However, as it turns out, good checkpointing algorithms do not set checkpoints very often. For all algorithms discussed in the remainder of this paper, each checkpoint is changed only $O(\log T)$ times up to time $T$.

\LV{\paragraph{The max-ratio discrepancy measure.} }
Each set of checkpoints, together with the current state and the state at time $0$, partitions the time from the process start to the current time $T$ into $k+1$ disjoint intervals. Clearly, without further problem-specific information, an ideal set of checkpoints would lead to all these intervals having identical length. Of course, this is not possible \LV{at all points in time }due to the restriction that new checkpoints can only be set on the current time. As discrepancy measure for a checkpointing algorithm, Ahlroth et~al.\ mainly regard the \emph{maximum gap ratio}, that is, the maximum ratio of the longest interval vs.\ the shortest interval (ignoring the last interval\LV{, which can be arbitrarily small}), over all current times $T$. They show that there is a simple algorithm achieving a discrepancy of two: Start with all checkpoints placed evenly, e.g., at times $1, \ldots, k$. At an even time $T$, remove one of the checkpoints at an odd time and place it at $T$. This will lead to all checkpoints being at the even times $2, 4, \ldots, 2k$ when $T = 2k$ is reached. Since these checkpoints form a scaled copy of the initial ones, we can continue in this fashion forever. It is easy to see that at all times, the intervals formed by neighboring checkpoints have at most two different lengths, the larger being twice the smaller\LV{ in case that not all lengths are equal}. This shows the discrepancy of two. 

\LV{It seems tempting to believe that one can do better, but, in fact, n}\SV{N}ot much improvement is possible for general $k$ as shown by the lower bound of $2^{1 - 1/\lceil(k+1)/2\rceil} = 2 (1 - o(1))$. For small values of $k$, namely $k = 2, 3, 4,$ and $5$, better upper bounds of approximately $1.414, 1.618, 1.755,$ and $1.755$, respectively, were shown.

\LV{\paragraph{The maximum distance discrepancy measure.} }In this work, we shall regard a different, and, as we find, more natural discrepancy measure. Recall that the \LV{actual }cost of reverting to a particular state is basically the cost of redoing the computation from the preceding checkpoint to the desired point in time. Adopting a worst-case view on the time to revert to, our aim is to keep the length of the longest interval small (at all times). Note that with time progressing, the interval lengths necessarily grow. Hence a fair point of comparison is the length $T / (k+1)$ of a longest interval in the (at time $T$) optimal partition \LV{of the time frame }into equal length intervals. For this reason, we say that a checkpointing algorithm (using $k$ checkpoints) has \emph{maximum distance discrepancy} (or simply discrepancy) $q$ if it places the checkpoints in such a way that at all times $T$, the longest interval has length at most $q T / (k+1)$. We denote by $q^*(k)$ the infimum discrepancy among all checkpointing algorithms using $k$ checkpoints. 

This \LV{maximum distance }discrepancy measure was suggested in~\cite{AhlrothPS11}. There it was remarked that an upper bound of $\beta$ for the gap-ratio \LV{discrepancy }implies an upper bound of $\beta (1 + \frac 1k)$ for the maximum distance discrepancy. Furthermore, for all $k$ an upper bound of $2$ and a lower bound of $1 + \frac 1k$  is shown for $q^*(k)$. For $k = 2, 3, 4,$ and $5$, stronger upper bounds of $1.785, 1.789, 1.624,$ and $1.565$, respectively, were shown.

\paragraph{Our results.} 
In this work, we show that the optimal discrepancy $q^*(k)$ is asymptotically bounded away from both one and two by a constant. We present algorithms that achieve a discrepancy of $1.59 + O(1/k)$ for all $k$ (Theorem~\ref{thm:algkarl}), and a discrepancy of $\ln(4) + o(1) \le 1.39 + o(1)$ for $k$ being any power of two (Theorem~\ref{thm:algjacub}). For small values of $k$, and this might be an interesting case in applications with memory-consuming states, we show superior bounds by suggesting a class of checkpointing algorithms and optimizing their parameters via a combination of exhaustive search and linear programming (Table~\ref{table:small_k_upper_bounds}). Experiments suggest $q^*(k) \le 1.7$ for all $k$ (Sect.~\ref{sec:linear_programming}). We complement these constructive results by a lower bound for $q^*(k)$ of $2 - \ln(2) - O(1/k) \ge 1.3 - O(1/k)$ (Theorem~\ref{thm:lower}). We round off this work with a natural, but seemingly nontrivial result: We show that for each $k$ there is indeed a checkpointing algorithm having discrepancy $q^*(k)$ (Theorem~\ref{thm:existence}). In other words, the infimum in the definition of $q^*(k)$ can be replaced by a minimum.

\SV{Due to space restrictions, some proofs are omitted. They can be found in the full version of this paper~\cite{arxiv_version}.}

\section{Notation and Preliminaries}\label{sec:notation}

\LV{In the checkpointing problem with $k$ checkpoints}\SV{In our setting}, we consider a long running computation during which we can choose to save the state at the current time $T$ in a checkpoint, or delete a previously placed one. We assume that our storage can hold at most $k$ checkpoints simultaneously, and that there are implicit checkpoints at time $t=0$ and the current time. We disregard any costs for placing or maintaining checkpoints. Consequently, we may assume that we only delete a previous checkpoint when a new one is placed.

An \emph{algorithm} for checkpoint placement can be described by two infinite sequences. First, the time points where new checkpoints are placed, i.e., a non-decreasing infinite sequence of reals $t_1 \le t_2 \le \ldots$ such that $\lim_{i \to \infty} t_i = \infty$, and second, a rule that describes which old checkpoints to delete when a new one is installed, that is, an injective function $d : [k+1..\infty) \to \N$ satisfying $d_i < i$ for all $i \ge k+1$. 

The algorithm $A$ described by $(t,d)$ will start with $t_1, \ldots, t_k$ as initial checkpoints and then for each $i \ge k+1$, at time $t_i$ remove the checkpoint at $t_{d_i}$ and set a new checkpoint at the current time $t_i$. We call the act of removing a checkpoint and placing a new one a \emph{step} of $A$. Note that there is little point in setting the first $k$ checkpoints to zero, so to make the following discrepancy measure meaningful, we shall always require that $t_k > 0$.

We call the set of checkpoints that exist at time $T$ \emph{active}. The\SV{se}\LV{ active checkpoints}, together with the two implicit checkpoints at times $0$ and $T$, define a sequence of $k+1$ interval lengths $\mathcal{L}_T = (\ell_0,\ldots,\ell_k)$. The \emph{discrepancy} $q(A,T)$ of an algorithm $A$ at time $T\ge t_k$ \LV{is a measure of how long the maximal interval is, normalized to be one if all intervals have the same length. It }is calculated as
\LV{\[
	q(A,T) := (k+1) \bar \ell_{T} / T,
\]}\SV{$  q(A,T) := (k+1) \bar \ell_{T} / T$, }%
where $\bar \ell_T = ||\mathcal{L}_T||_{\infty}$ denotes the length of the longest interval.\LV{ We also use the term \emph{discrepancy} when we refer to the scaled length of a single interval.}

The discrepancy $\Perf(A)$ of an algorithm $A$ then is the supremum over the discrepancy over all times $T$, i.e.,
\[
	\Perf(A) := \sup_{T \ge t_k} q(A,T).
\]
Hence the discrepancy of an algorithm would be $1$, if it \SV{always }kept its checkpoints evenly distributed\LV{ at all times}. Denote the infimum discrepancy of a checkpointing algorithm using $k$ checkpoints by
\[
  q^*(k) := \inf_A \Perf(A),
\] 
where $A$ runs over all algorithms using $k$ checkpoints. We will see in Sect.~\ref{sec:existence} that algorithms achieving this discrepancy actually exist.

Note that we allow checkpointing algorithms to set checkpoints at continuous time points. One can convert any such algorithm to an algorithm with integral checkpoints by rounding all checkpointing times $t_i$ down. This does not increase the discrepancy since $\lfloor t_i \rfloor - \lfloor t_{i-1} \rfloor \le t_i-t_{i-1}+1$, but with discrete time there are at most $\lfloor t_i \rfloor - \lfloor t_{i-1} \rfloor -1$ steps to recompute in this interval.

In the definition of the discrepancy, the supremum is never attained at some $T$ with $t_i< T < t_{i+1}$ for any $i$, \SV{further, at any time $t_i$ it suffices to consider the two newly created intervals by deleting and storing a checkpoint. This is made precise in the following two lemmas.}\LV{ as shown in the following lemma.}

\begin{lemma}\label{lem:steps_are_enough}
  In the definition of the discrepancy it suffices to consider times $T = t_i$ for all $i \ge k$, i.e., we have
  \[
    \Perf(A) = \sup_{i \ge k} q(A,t_i).
  \]
\end{lemma}
\LV{
\begin{proof}
  Consider a time $T$ with $t_i< T < t_{i+1}$ for any $i \ge k$. We show that 
  \[
    q(A,T) \le \max\{ q(A,t_i), q(A,t_{i+1}) \}.
  \]
  Denote the active checkpoints at time $T$ by $x_1,\ldots,x_k$. Note that $x_k = t_i$, since $t_i$ was the last time we set a checkpoint. 
  Consider the interval $[x_k,T]$. Its discrepancy is exactly
  \[
    (k+1) \frac{T-x_k}{T} \le (k+1) \frac{t_{i+1} - x_k}{t_{i+1}} \le q(A,t_{i+1}).
  \]
  Any other interval at time $T$ is of the form $[x_{j-1},x_j]$ for some $1 \le j \le k$ (where we set $x_0 := 0$), whose discrepancy is
  \[
    (k+1) \frac{x_j - x_{j-1}}{T} \le (k+1) \frac{x_j - x_{j-1}}{t_i} \le q(A,t_i).
  \]
  Together, this proves the claim.~
\end{proof}

To bound the discrepancy of an algorithm we need to bound the largest of the $q(A,t_i)$ over all $i \ge k$. For this purpose, it suffices to look at the two newly created intervals at time $t_i$ for each $i$, as made explicit by the following lemma. 

}
\begin{lemma}\label{lem:only_new_intervals} 
Let $i > k$ and let $\ell_1, \ell_2$ be the lengths of the two newly created intervals at time $t_i$ due to the removal and the insertion of a checkpoint. Then
\[
  \max\{q(A,t_{i-1}),q(A,t_i)\} = \max\{ q(A,t_{i-1}), (k+1)\ell_1/t_i, (k+1)\ell_2/t_i\}.
\]
\end{lemma}
\LV{
\begin{proof} 
If $\ell_1$ or $\ell_2$ is the longest interval at time $t_i$ the claim holds. Any other interval existed already at time $t_{i-1}$ and had a larger discrepancy at this time, as we divide by the current time to compute the discrepancy. Thus, if any other interval is the longest at time $t_i$, then we have $q(A,t_{i-1}) \ge q(A,t_i)$ and the claim holds again.~
\end{proof}
}
\SV{Intuitively, both lemmas rely on the fact that unchanged intervals improve in discrepancy as time progresses. In Lemma~\ref{lem:only_new_intervals}, all intervals improve in discrepancy, except the two that change. In Lemma~\ref{lem:steps_are_enough}, the interval $[t_n,T]$, defined by the last checkpoint $t_n$ and the current time $T$, has maximal discrepancy when we place a new checkpoint.}

Often, it will be useful to use a different notation for the checkpoint \LV{that is }removed in step $i$. Instead of  the global index $d$, one can also use the index $p: [k+1..\infty) \to [1..k]$ among the active checkpoints, i.e.,
\[p_i = d_i - |\{j\in [i-1] \,|\, d_j<d_i\}|.\]

We call an algorithm $A=(t,p)$ \emph{cyclic}, if the $p_i$ are periodic with some period $n$, i.e., $p_i=p_{i+n}$ for all $i$, and after $n$ steps $A$ has transformed the intervals to a scaled version of themselves, that is, $\mathcal{L}_{t_{k+jn}} = \gamma^j \mathcal{L}_{t_k}$ for some $\gamma>1$ and all $j\in \mathbb{N}$. We call $\gamma$ the \emph{scaling factor}. For a cyclic algorithm $A$, it suffices to fix the \emph{pattern} of removals $P=(p_{k+1},\ldots,p_{k+n})$ and the checkpoint positions $t_1,\ldots,t_k,t_{k+1},\ldots,t_{k+n}$. Since our discrepancy notion is invariant under scaling, we can assume without loss of generality that $t_k=1$ (and hence $t_{k+n}=\gamma$).

Since cyclic algorithms transform the starting position to a scaled copy of itself, it is easy to see that their discrepancy is given by the maximum over the discrepancies during one period, i.e., for cyclic algorithms $A$ with period $n$ we have
\[
	\Perf(A) = \max_{k < i \leq k+n} q(A,t_i).
\] 
This makes this class of algorithms easy to analyze.

\section{Introductory Example -- A Simple Bound for $k=3$}\label{sec:golden_ratio}
For the case of $k=3$ there is a very simple algorithm, \AlgSimple, with a discrepancy of $4/{\phi^2}\approx 1.53$, where $\phi = (\sqrt{5}+1)/2$ is the golden ratio. \LV{Because the algorithm is so simple, w}\SV{W}e use it to familiarize ourselves with the notation \LV{we introduced in }\SV{from }Sect.~\ref{sec:notation}.
The algorithm is cyclic with a pattern of length one. We prove the following theorem.

\begin{theorem} For $k=3$ there is a cyclic algorithm $\AlgSimple$ with period length one and 
\[
	\Perf(\AlgSimple)=\LV{\frac{4}{\phi^2}}\SV{4/\phi^2}.
\]
\end{theorem}
\begin{proof}
We fix the pattern to be $P=(1)$, that is, algorithm \AlgSimple\ always removes the oldest checkpoint. For this simple pattern it is easy to calculate the discrepancy depending on the scaling factor $\gamma$. Since the intervals need to be a scaled copy of themselves after just one step and we can fix $t_3=1$, we know immediately that
\[
	t_1=\SV{\gamma^{-2}}\LV{\frac{1}{\gamma^2}},\ t_2=\SV{\gamma^{-1}}\LV{\frac{1}{\gamma}},\ t_3=1,\ t_4=\gamma,
\]
and hence the discrepancy is determined by
\LV{
\[
    4\cdot\max\left\{ \frac{t_1-0}{t_3}, \frac{t_2-t_1}{t_3}, \frac{t_3-t_2}{t_3} \right\} =
	4\cdot\max \left\{\frac{1}{\gamma^2}, \frac{\gamma-1}{\gamma^2}, \frac{\gamma-1}{\gamma}\right\}.
\]
} 
\SV{
\[
    4\cdot\max\left\{ t_1/t_3, (t_2-t_1)/t_3, (t_3-t_2)/t_3 \right\} =
	4\cdot\max \left\{\gamma^{-2}, (\gamma-1)/\gamma^2, (\gamma-1)/\gamma\right\}.
\]
}
\LV{
Since $\gamma>1$, the second term is always smaller than the third and can be ignored. As $1/\gamma^2$ is decreasing and $(\gamma-1)/\gamma$ is increasing, the maximum is minimal when they are equal. Simple calculation shows this to be the case at $\gamma = \phi$.}
\SV{A simple calculation shows this maximum to be minimal at $\gamma=\phi$.}

Hence for $k=3$ the algorithm with pattern $(1)$ and checkpoint positions $t_1=1/\phi^2$, $t_2=1/\phi$, $t_3= 1$, and $t_4=\phi$ has discrepancy $4/\phi^2 \approx 1.53$. ~
\end{proof}
\LV{
The experiments in Sect.~\ref{sec:linear_programming} indicate that for $k=3$ this is optimal among all cyclic algorithms with a period of length at most 6. }

\section{A Simple Upper Bound for Large $k$}\label{sec:karls_algorithm}
In this section we present an algorithm, \AlgKarl, with a discrepancy of roughly $1.59$ for large $k$. This improves upon the asymptotic bound of 2 from~\cite{AhlrothPS11}. Moreover, \AlgKarl\ is easily implemented for all $k$.

Like the algorithm \AlgSimple\ of the previous section, the algorithm \AlgKarl\ is cyclic. It has a simple pattern of length $k$. The pattern is just $(1,\ldots,k)$, that is, at the $i$-th step of a period \AlgKarl\ deletes the $i$-th active checkpoint. Overall, during one period \AlgKarl\ removes all checkpoints at times $t_i$ with odd index $i$, as shown in Fig.~\ref{fig:karls_algorithm_movement}. 

This removal pattern is identical to the one of \textsc{Powers-Of-Two} algorithm from~\cite{AhlrothPS11}. However, that algorithm starts with a uniform checkpoint distribution where removing any checkpoint doubles the maximum interval. This leads to an asymptotic discrepancy of two. In contrast, \AlgKarl\  places checkpoints on a polynomial. For $i\in [1,2k]$ we set $t_i=(i/k)^\alpha$, where $\alpha$ is a constant. In the analysis we optimize the choice of $\alpha$ and set $\alpha := 1.302$. \LV{For this algorithm w}\SV{W}e show the following theorem.
\begin{theorem} \label{thm:algkarl}
Algorithm \AlgKarl\ has a discrepancy of at most \SV{$1.586 + O(k^{-1})$.}
\LV{\[\Perf(\AlgKarl)\leq 1.586 + O(k^{-1}).\]}
\end{theorem}
\LV{%
Experiments show that the discrepancy of algorithm \AlgKarl\ is close to the bound of $1.586$ even for moderate sizes of $k$. Comparisons using the optimization method from Sect.~\ref{sec:linear_programming} indicate that for the pattern $(1,\ldots,k-1)$ of algorithm \AlgKarl, different checkpoint placements can yield only improvements of about 4.5\% for large $k$. Experimental results are summarized in Fig.~\ref{fig:karls_algorithm_performance}.
}
\SV{%
  The proof of the above theorem is a straightforward calculation of the lengths of the newly created intervals at each step, similar to the previous section, together with some analytic estimations of the resulting discrepancies.
}
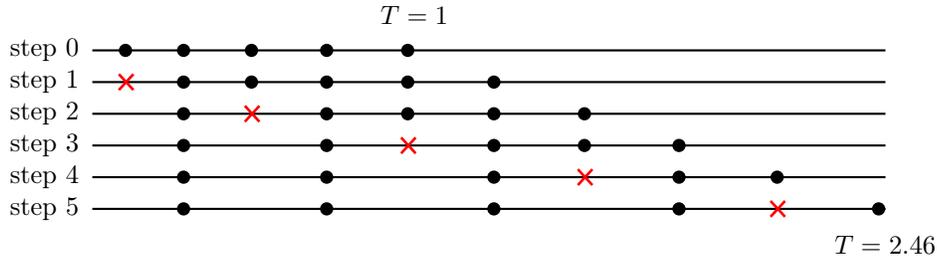
\begin{figure}
  \centering
  \definecolor{gray}{RGB}{178,178,178}
    \begin{tikzpicture}[y=1pt, x=1pt,yscale=-1, inner sep=0pt, outer sep=0pt]
    \def\linespread{\baselineskip}
    \foreach \i in {0,1,2,...,5} {
      \path[shift={(0,0)}, draw=black, line join=miter, line cap=butt, line width=0.8pt]
        (0,\i * \linespread) node[left=5] {step \i} -- (300, \i * \linespread) ;
    }
    \foreach \x in {15.031959126760265, 37.01302500108005, 62.7008195354366, 91.13675790215085, 121.8082186} {
      \path[fill=black]
              (\x,0)arc(0.000:180.000:2.500)arc(-180.000:0.000:2.500) -- cycle;
    }
    \draw (121.8082186,0) node [above=10] {$T=1$};

    \foreach \x in {37.01302500108005, 62.7008195354366, 91.13675790215085, 121.8082186, 154.38752736640112} {
      \path[fill=black]
              (\x,\linespread)arc(0.000:180.000:2.500)arc(-180.000:0.000:2.500) -- cycle;
    }
    \foreach \x in {37.01302500108005, 91.13675790215085, 121.8082186, 154.38752736640112, 188.64401807637188} {
      \path[fill=black]
              (\x,2* \linespread)arc(0.000:180.000:2.500)arc(-180.000:0.000:2.500) -- cycle;
    }
    \foreach \x in {37.01302500108005, 91.13675790215085, 154.38752736640112, 188.64401807637188, 224.40502068320237} {
      \path[fill=black]
              (\x,3* \linespread)arc(0.000:180.000:2.500)arc(-180.000:0.000:2.500) -- cycle;
    }
    \foreach \x in {37.01302500108005, 91.13675790215085, 154.38752736640112, 224.40502068320237, 261.5356213560098} {
      \path[fill=black]
              (\x,4 * \linespread)arc(0.000:180.000:2.500)arc(-180.000:0.000:2.500) -- cycle;
    }
    \foreach \x in {37.01302500108005, 91.13675790215085, 154.38752736640112, 224.40502068320237, 299.92701565777264} {
      \path[fill=black]
              (\x,5* \linespread)arc(0.000:180.000:2.500)arc(-180.000:0.000:2.500) -- cycle;
    }
    \draw (299.927, 5*\linespread) node [below=10pt] {$T=2.46$};

    \foreach \i/\x in {1 / 15.0319591268 , 2 / 62.7008195354 , 3 / 121.8082186 , 4 / 188.644018076 , 5 / 261.535621356 } {
      \path[draw=red, line width = 1pt] 
        (\x - 5, \i * \linespread - 3) -- (\x + 0.5 , \i  * \linespread + 3);
      \path[draw=red, line width = 1pt]
        (\x - 5, \i * \linespread + 3) -- (\x + 0.5, \i  * \linespread- 3);
    }
  \end{tikzpicture}
\caption{One period of the algorithm \AlgKarl\ from Sect.~\ref{sec:karls_algorithm} for $k=5$. After one period all intervals are scaled by the same factor.}
\label{fig:karls_algorithm_movement}
\end{figure}
\LV{
  \begin{proof}
  As algorithm \AlgKarl\ is cyclic, we can again compute the discrepancy from the $2k$ checkpoint positions and the pattern,
  \begin{align*}
    \Perf(\AlgKarl) = \max_{k < i \leq 2k} (k+1) \bar \ell_{t_i} / t_i,
  \end{align*} 
  where $\bar \ell_{t_i}$ is the length of the longest interval at time $t_i$.
  By Lemma~\ref{lem:only_new_intervals} it suffices to consider newly created intervals at times $t_{k+1},\ldots, t_{2k}$. Note that at time $t_i$ we create the intervals $[t_{i-1},t_i]$ (from insertion of a checkpoint at $t_i$) and $[t_{2(i-k)-2},t_{2(i-k)}]$ (from deletion of the checkpoint at $t_{2(i-k)-1}$). The discrepancy of the new interval by insertion is, for $k < i \le 2k$, 
  \begin{align*}
    (k+1) \frac{t_i - t_{i-1}}{t_i} = (k+1) \frac{i^\alpha - (i-1)^\alpha}{i^\alpha} 
    \le (k+1) \frac{(k+1)^\alpha - k^\alpha}{(k+1)^\alpha}.
  \end{align*}
  Using $(x+1)^c - x^c \le c (x+1)^{c-1}$ for any $x \ge 0$ and $c \ge 1$, this simplifies to
  \begin{align*}
    \le (k+1) \frac{ \alpha (k+1)^{\alpha-1} }{ (k+1)^\alpha } = \alpha,
  \end{align*}
  for any constant $\alpha \ge 1$.

  For the new interval from deleting the checkpoint at $t_{2(i-k)-1}$ we get a discrepancy of 
  \begin{align*}
    (k+1) \frac{t_{2(i-k)} - t_{2(i-k)-2}}{t_i} &= (k+1) \frac{(2(i-k))^\alpha - (2(i-k)-2)^\alpha}{i^\alpha}  \\
    &\le (k+1) 2^\alpha \frac{ \alpha (i-k)^{\alpha-1} }{ i^\alpha },
  \end{align*}
  where we used again $(x+1)^c - x^c \le c (x+1)^c$. An easy computation shows that $(i-k)^{\alpha-1} / i^\alpha$ is maximized at $i = \alpha k$ over $k < i \le 2k$. Hence, we can upper bound this discrepancy by
  \begin{align*}
    \le \Big(1+ \frac 1k \Big) 2^\alpha \frac{ \alpha (\alpha - 1)^{\alpha - 1} }{ \alpha^\alpha }
    = 2^\alpha \Big(1- \frac 1\alpha \Big)^{\alpha-1} + O(k^{-1}).
  \end{align*}
  We optimize the latter term numerically and obtain for $\alpha = 1.302$ an upper bound of
  \begin{align*}
    1.586 + O(k^{-1}).
  \end{align*}
  Note that this bound is larger than the bound $\alpha = 1.302$ from the new intervals from insertion. Hence, overall we get the desired upper bound.~
  \end{proof}
}

\section{An Improved Upper Bound for Large $k$}\label{sec:jakubs_algorithm}
In this section we present the algorithm \AlgJakub\ that yields a discrepancy of roughly $\ln(4) \approx 1.39$ for large $k$. Compared to the algorithm \AlgKarl\ from the last section, \AlgJakub\ has a considerably better discrepancy at the price of a more involved analysis, and it only works for $k$ being a power of two. 

\begin{theorem} \label{thm:algjacub}
  For $k\ge 8$ being any power of 2, the algorithm \AlgJakub\ has discrepancy
\LV{\[    
  \Perf(\AlgJakub) \le \ln(4) + \frac{0.05}{\lg(k/4)} + O\Big(\frac1{k}\Big).
  \]
}\SV{\[
  \Perf(\AlgJakub) \le \ln(4) + 0.05/\lg(k/4) + O\left(k^{-1}\right)
\]}
\end{theorem}
Here and in the remainder of this paper, let `$\lg$' denote the binary and `$\ln$' the natural logarithm. Note that the term $O(1/k)$ quickly tends to 0, whereas the $\Theta(1/\lg(k/4))$ term is small due to the constant $0.05$. Hence, this discrepancy is close to $\ln(4)$ already for moderate $k$. Also note that $\ln(4)$ is by less than $0.1$ larger than our lower bound from Sect.~\ref{sec:lower_bound}, leaving room for less than a $6\%$ improvement over\LV{ the upper bound for} algorithm \AlgJakub\ for large $k$. 
\LV{We verified experimentally that algorithm \AlgJakub\ yields very good bounds already for relatively small $k$. The results are summarized in Fig.~\ref{fig:jakubs_algorithm_performance}.}

\subsection{The Algorithm \AlgJakub} 

The initial checkpoints $t_1,\ldots,t_k$ satisfy the equation
\begin{equation} \label{eq:tialpha}
  t_i = \alpha t_{i/2} 
\end{equation}
for each even $1 \le i \le k$ and some $\alpha = \alpha(k) \ge 2$. Precisely, we set
\[
  \alpha := 2^{1 + \frac{\lg(\sqrt{2}/\ln 4)}{\lg(k/4)} } \approx 2^{1+ \frac{0.029}{\lg(k/4)}}.
\]
However, the usefulness of this expression becomes clear only in the analysis\LV{ of the algorithm}.

During one period we delete all odd checkpoints $t_1,t_3,\ldots,t_{k-1}$ and insert\LV{ the new checkpoints}
\begin{equation} \label{eq:newti}
  t_{k+i} := \alpha t_{k/2+i},
\end{equation}
for $1 \le i \le k/2$. Then after one period we end up with the checkpoints
\[
\begin{array}{l*{5}{l@{,\,}}@{\quad}*{3}{l@{,\ }}ll}
  &(t_2 & t_4 & \ldots\, & t_{k-2} & t_k & t_{k+1} & t_{k+2} &\ldots\, & t_{k+k/2})&  \\
  = \alpha \cdot&(t_1 & t_2 & \ldots\, & t_{k/2-1} &t_{k/2} & t_{k/2+1} & t_{k/2+2} & \ldots\, &t_{k/2+k/2})& = \alpha (t_1,t_2,\ldots,t_k),
\end{array}
\]
which proves cyclicity. Note that~\eqref{eq:tialpha} and~\eqref{eq:newti} allow us to compute all $t_i$ from the values $t_{k/2+1},\ldots,t_k$, however, we still have some freedom to choose the latter values. Without loss of generality we can set $t_k := 1$, then $t_{k/2} = \alpha^{-1}$. In between these two values, we interpolate $\lg t_i$ linearly, i.e., we set for $i \in (k/2,k]$
\begin{equation} \label{eq:tiinterpol}
  t_i := \alpha^{2i/k - 2},
\end{equation}
completing the definition of the $t_i$. Note that\LV{ this equation}\SV{\eqref{eq:tiinterpol}} also works for $i=k$ and $i=k/2$.

\LV{There is one more freedom we have with this algorithm, namely in which order we delete all odd checkpoints during one period, i.e., we need to fix the pattern of removals.}
In iteration $1 \le i \le k/2$ we insert the checkpoint $t_{k+i}$ and remove the checkpoint $t_{d(i+k)}$, defined as follows. For $m \in \N = \N_{\ge 1}$ let $2^{e(m)}$ be the largest power of 2 that divides $m$. We define $S \colon \N \to \N, S(m) := m / 2^{e(m)}$. Note that $S(m)$ is an odd integer. Using this definition, we set
\LV{\begin{equation} \label{eq:defd}
  d(k+i) := S\Big(i+ \frac k2 \Big),
\end{equation}}
\SV{\begin{equation}\label{eq:defd}
  d(k+i) := S(i+k/2)
\end{equation}}
finishing the definition of the algorithm \AlgJakub. If we write this down as a pattern, then we have $p_i = 1 + k/(2^{1+e(i)})$ for $1 \le i < k/2$ and $p_{k/2} = 1$.
For intuition as to the behavior of this pattern, see the example in Fig.~\ref{fig:jakubs_algorithm_movement}.
\LV{
The following lemma implies that the deletion behavior of \AlgJakub\ is indeed well-defined, meaning that during one period we delete all odd checkpoints $t_1,t_3,\ldots,t_{k-1}$ (and no point is deleted twice).
\begin{lemma}
  The function $S$ induces a bijection between $\{k/2 < i \le k\}$ and $\{1 \le i \le k \mid i \text{ is odd} \}$.
\end{lemma}
\begin{proof}
  Let $A := \{k/2 < i \le k\}$ and $B := \{1 \le i \le k \mid i \text{ is odd} \}$.
  Since $S(m) \le m$ and $S(m)$ is odd for all $m \in \N$, we have $S(A) \subseteq B$. Moreover, $A$ and $B$ are of the same size. We present an inverse function to finish the proof. Let $x \in B$. Note that there is a unique number $y \in \N$ such that $x 2^y \in A$, since $A$ is a range between two consecutive powers of 2 and $x \le k$. Setting $S^{-1}(x) = x 2^y$ we have found the inverse. ~
\end{proof}
}\SV{It is not hard to see that with $d$ as defined above, we indeed delete all odd checkpoints $t_1,t_3,\ldots,t_{k-1}$ during one period,
in the full version of this paper we include a formal proof of this.}
\begin{figure}
  \centering
  \definecolor{gray}{RGB}{178,178,178}
    \begin{tikzpicture}[y=1pt, x=1pt,yscale=-1, inner sep=0pt, outer sep=0pt]
    \def\linespread{\baselineskip}
    \foreach \i in {0,1,2,...,8} {
      \path[shift={(0,0)}, draw=black, line join=miter, line cap=butt, line width=0.8pt]
        (0,\i * \linespread) node[left=5] {step \i} -- (303, \i * \linespread) ;
    }
    \foreach \x in {9.086718622889968, 18.28737691914821, 25.94320242291756, 36.80406189100548, 43.83609267359075, 52.211710397032434, 62.18762978904628, 74.06961521412411, 80.83661709068994, 88.22185242594281, 96.28180304394778, 105.07811094961241, 114.67804976293608, 125.155039223504, 136.58920670012435, 149.068} {
      \path[fill=black]
              (\x,0)arc(0.000:180.000:2.500)arc(-180.000:0.000:2.500) -- cycle;
    }
    \draw (149.068,0) node [above=10] {$T=1$};

    \foreach \x in {9.086718622889968, 18.28737691914821, 25.94320242291756, 36.80406189100548, 43.83609267359075, 52.211710397032434, 62.18762978904628, 74.06961521412411, 88.22185242594281, 96.28180304394778, 105.07811094961241, 114.67804976293608, 125.155039223504, 136.58920670012435, 149.068, 162.6868561641611} {
      \path[fill=black]
              (\x,\linespread )arc(0.000:180.000:2.500)arc(-180.000:0.000:2.500) -- cycle;
    }
    \foreach \x in {9.086718622889968, 18.28737691914821, 25.94320242291756, 36.80406189100548, 52.211710397032434, 62.18762978904628, 74.06961521412411, 88.22185242594281, 96.28180304394778, 105.07811094961241, 114.67804976293608, 125.155039223504, 136.58920670012435, 149.068, 162.6868561641611, 177.549931364065} {
      \path[fill=black]
              (\x,2*\linespread)arc(0.000:180.000:2.500)arc(-180.000:0.000:2.500) -- cycle;
    }
    \foreach \x in {9.086718622889968, 18.28737691914821, 25.94320242291756, 36.80406189100548, 52.211710397032434, 62.18762978904628, 74.06961521412411, 88.22185242594281, 105.07811094961241, 114.67804976293608, 125.155039223504, 136.58920670012435, 149.068, 162.6868561641611, 177.549931364065, 193.7708974815676} {
      \path[fill=black]
              (\x,3*\linespread)arc(0.000:180.000:2.500)arc(-180.000:0.000:2.500) -- cycle;
    }
    \foreach \x in {9.086718622889968, 18.28737691914821, 36.80406189100548, 52.211710397032434, 62.18762978904628, 74.06961521412411, 88.22185242594281, 105.07811094961241, 114.67804976293608, 125.155039223504, 136.58920670012435, 149.068, 162.6868561641611, 177.549931364065, 193.7708974815676, 211.47381146446045} {
      \path[fill=black]
              (\x,4*\linespread)arc(0.000:180.000:2.500)arc(-180.000:0.000:2.500) -- cycle;
    }
    \foreach \x in {9.086718622889968, 18.28737691914821, 36.80406189100548, 52.211710397032434, 62.18762978904628, 74.06961521412411, 88.22185242594281, 105.07811094961241, 125.155039223504, 136.58920670012435, 149.068, 162.6868561641611, 177.549931364065, 193.7708974815676, 211.47381146446045, 230.79406410635147} {
      \path[fill=black]
              (\x,5*\linespread)arc(0.000:180.000:2.500)arc(-180.000:0.000:2.500) -- cycle;
    }
    \foreach \x in {9.086718622889968, 18.28737691914821, 36.80406189100548, 52.211710397032434, 74.06961521412411, 88.22185242594281, 105.07811094961241, 125.155039223504, 136.58920670012435, 149.068, 162.6868561641611, 177.549931364065, 193.7708974815676, 211.47381146446045, 230.79406410635147, 251.87941550709863} {
      \path[fill=black]
              (\x,6*\linespread)arc(0.000:180.000:2.500)arc(-180.000:0.000:2.500) -- cycle;
    }
    \foreach \x in {9.086718622889968, 18.28737691914821, 36.80406189100548, 52.211710397032434, 74.06961521412411, 88.22185242594281, 105.07811094961241, 125.155039223504, 149.068, 162.6868561641611, 177.549931364065, 193.7708974815676, 211.47381146446045, 230.79406410635147, 251.87941550709863, 274.8911251329348} {
      \path[fill=black]
              (\x,7*\linespread)arc(0.000:180.000:2.500)arc(-180.000:0.000:2.500) -- cycle;
    }
    \foreach \x in {18.28737691914821, 36.80406189100548, 52.211710397032434, 74.06961521412411, 88.22185242594281, 105.07811094961241, 125.155039223504, 149.068, 162.6868561641611, 177.549931364065, 193.7708974815676, 211.47381146446045, 230.79406410635147, 251.87941550709863, 274.8911251329348, 300.0051851189134} {
      \path[fill=black]
              (\x,8*\linespread)arc(0.000:180.000:2.500)arc(-180.000:0.000:2.500) -- cycle;
    }

    \draw (300.00518511, 8*\linespread) node [below=10pt] {$T=2.012$};

    \foreach \i/\x in {1/80.8366170907, 2/43.8360926736, 3/96.2818030439, 4/25.9432024229, 5/114.678049763, 6/62.187629789, 7/136.5892067, 8/9.08671862289} {
      \path[draw=red, line width = 1pt] 
        (\x - 5, \i * \linespread - 3) -- (\x + 0.5 , \i  * \linespread + 3);
      \path[draw=red, line width = 1pt]
        (\x - 5, \i * \linespread + 3) -- (\x + 0.5, \i  * \linespread - 3);
    }
    \foreach \x in {9.086718622889968, 18.28737691914821, 36.80406189100548, 74.06961521412411, 149.068, 300.0051851189134} {
      \path[draw=black, dotted, line width = 0.8pt]
        (\x -2.5, -7) -- (\x -2.5, 8*\linespread + 9);
    }
%
  \end{tikzpicture}
\caption{One period of the algorithm \AlgJakub\ for $k=16$. Note that, recursively, checkpoints are removed twice as often from the right half of the initial setting \LV{(at steps $i$ where $i \mod 2=1$) }as from the second quarter. }
\label{fig:jakubs_algorithm_movement}
\end{figure}
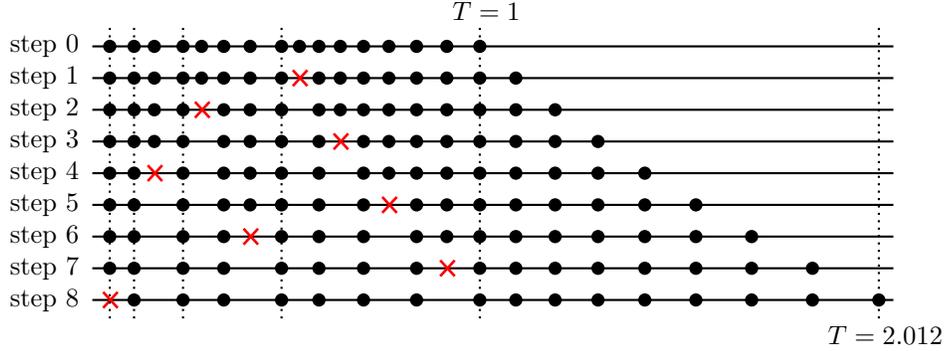

\subsection{Discrepancy Analysis} 
We now bound the largest discrepancy encountered during one period, i.e.,
\[
  \Perf(\AlgJakub) = \max_{1 \le i \le k/2} q(\AlgJakub,t_{i+k}) = (k+1) \max_{1 \le i \le k/2} \overline{\ell}_{t_{i+k}} / t_{i+k}.
\]
\SV{

  By Lemma~\ref{lem:only_new_intervals}, we only have to consider intervals newly created by insertion and deletion at any step. We do this exemplarily for the intervals from insertion.
}\LV{

  We first compute the maximum and later multiply with the factor $k+1$.
  By Lemma~\ref{lem:only_new_intervals}, we only have to consider intervals newly created by insertion and deletion at any step.

}
\paragraph{Intervals from Insertion:}
We \LV{first }compute the discrepancy of the interval newly added at time $t_{i+k}$, $1 \le i \le k/2$. Its length is $t_{i+k} - t_{i+k-1}$, so its discrepancy \LV{(without the factor $k+1$) }is
  \begin{align*}
    \SV{(k+1)} \frac{t_{i+k} - t_{i+k-1}}{t_{i+k}}
    &= \SV{(k+1) \left(} 1 - \frac{t_{i+k-1}}{t_{i+k}} \SV{\right)}  \\
    &= \SV{(k+1) \left(} 1 - \frac{ t_{i+k/2-1} }{ t_{i+k/2} } \SV{\right)} \LV{\\}
    \LV{&}\stackrel{\eqref{eq:tiinterpol}}{=} \SV{(k+1) \left(} 1 - \alpha^{-2/k} \SV{\right)},
  \end{align*}
  where the second equality holds because of \eqref{eq:newti} if $i>1$ or \eqref{eq:tialpha} if $i=1$.

  Using $e^x \ge 1+x$ for $x \in \mathbb{R}$ yields a bound on the discrepancy of 
  \[
    \SV{(k+1) }\frac{t_{i+k} - t_{i+k-1}}{t_{i+k}} \le \SV{(k+1)} \ln(\alpha) \frac{2}{k} = \ln(\alpha^2)\SV{ + O(k^{-1})}.
  \]
 \SV{ Since we choose $\alpha \approx 2$, we obtain a discrepancy of roughly $\ln(4)$, and the error term can easily be seen to be bounded by $0.05/\ln(k/4) + O(k^{-1})$. }

\LV{

\paragraph{Deleting $t_1$:} 
We show similar bounds for the intervals we get from deleting an old checkpoint. We first analyze the deletion of $t_1$---this case is different from the general one, since $t_1$ has no predecessor. Note that $t_1$ is deleted at time $t_{3k/2}$. The deletion of $t_1$ creates the interval $[0,t_2]$. This interval has discrepancy
\begin{align*}
  \frac{t_2}{t_{3k/2}} 
  \stackrel{\eqref{eq:newti},\eqref{eq:tialpha}}{=} \frac{\alpha t_1}{\alpha t_k}
  \stackrel{\eqref{eq:tialpha}}{=} \alpha^{-\lg k} \le 1/k,
\end{align*}
since we choose $\alpha \ge 2$. Hence, this discrepancy is dominated by the one we get from newly inserted intervals.

\paragraph{Other Intervals from Deletion:} 
It remains to analyze the discrepancy of the intervals we get from deletion in the general case, i.e., at some time $t_{i+k}$, $1 \le i < k/2$. At this time we delete checkpoint $d(i+k)$, so we create the interval $[t_{d(i+k)-1},t_{d(i+k)+1}]$ of discrepancy
\[
  q_i := \frac{ t_{d(i+k)+1} - t_{d(i+k)-1} }{ t_{i+k} }
  \stackrel{\eqref{eq:newti},\eqref{eq:defd}}{=} \frac{ t_{S(i+k/2)+1} - t_{S(i+k/2)-1} }{ \alpha t_{i+k/2} }.
\]
Let $h := e(i+k/2)$, so that $2^h$ is the largest power of 2 dividing $i+k/2$, and $2^h \,S(i+k/2) = i+k/2$. Then $t_{S(i+k/2)+1} = \alpha^{-h} t_{i+k/2+2^h}$ by \eqref{eq:tialpha}, and a similar statement holds for $t_{S(i+k/2)-1}$, yielding
\[
  q_i = \alpha^{-1-h} \frac{ t_{i+k/2+2^h} - t_{i+k/2-2^h} }{ t_{i+k/2} }.
\]
Using \eqref{eq:tiinterpol} we get $t_{i+k/2} = \alpha^{2i/k - 1}$. Comparing this with the respective terms for $t_{i+k/2+2^h}$ and $t_{i+k/2-2^h}$ yields
\begin{align*}
  q_i &= \alpha^{-1-h} \left(\alpha^{2^{h+1}/k} - \alpha^{-2^{h+1}/k}\right) \\
  &= \alpha^{-1-h} \cdot 2 \sinh \left( \ln \left(\alpha^2\right) 2^h / k \right).
\end{align*}
By elementary means one can show that the function $f(x) = x^{-A} \sinh(B x)$, $A \ge 1, B > 0$, is convex on $\mathbb{R}_{\ge 0}$. Since convex functions have their maxima at the boundaries of their domain, and since by above equation $q_i$ can be expressed using $f(2^h)$ (for $A = \lg \alpha$ and $B = \ln(\alpha^2)/k$), we see that $q_i$ is maximal at (one of) the boundaries of $h$. Recall that we treated $i=k/2$ separately, and observe that the largest power of 2 dividing $i+k/2$, $1 \le i < k/2$ is at most $k/4$. Hence, we have $0 \le 2^h \le k/4$ and 
\[
  q_i \le \max\left\{ 2 \alpha^{-1} \sinh(\ln(\alpha^2)/k), 2 \alpha^{-1} (k/4)^{-\lg \alpha} \sinh( \ln(\alpha)/2) \right\}.
\]
We simplify using $\alpha \ge 2$ and $\sinh(x) = x + O(x^2)$ to get
\begin{equation} \label{eq:jakubsqi}
  q_i \le \max\left\{ \ln(\alpha^2)/k + O(1/k^2), (k/4)^{-\lg \alpha} \sinh( \ln(\alpha)/2) \right\}.
\end{equation}
The first term is already of the desired form. For the second one, note that setting $\alpha = 2$ we would get a discrepancy of $4 \sinh(\ln(2)/2) / k = \sqrt{2}/k$. We get a better bound by choosing
\[
  \alpha := 2^{1 + \frac{c}{\lg(k/4)} },
\]
with $c := \lg( \sqrt{2} / \ln(4)) \approx 0.029$.
Then the second bound on $q_i$ from above becomes
\[
  (k/4)^{-\lg \alpha} \sinh( \ln(\alpha)/2) = \frac 4k 2^{-c} \sinh\left( \frac{\ln(2)}{2} \Big(1+\frac{c}{\lg(k/4)}\Big)\right).
\]
The particular choice of $c$ allows to bound the derivative of $\sinh((1+x) \ln(2)/2)$ for $x \in [0,c]$ from above by 
\[
  \frac{\ln(2)}{2} \cosh((1+c) \ln(2)/2) < 0.39.
\]
Hence, we can upper bound 
\[
  \sinh\left( \frac{\ln(2)}{2} \Big(1+\frac{c}{\lg(k/4)}\Big)\right) \le \sinh(\ln(2)/2) + \frac{0.39 c}{\lg(k/4)}.
\]
Thus, in total the second bound on $q_i$ from inequality~(\ref{eq:jakubsqi}) becomes 
\[
  (k/4)^{-\lg \alpha} \sinh( \ln(\alpha)/2) 
  \le  \frac 4k 2^{-c} \sinh(\ln(2)/2) + \frac{4 \cdot 2^{-c} \cdot 0.39 c}{k \lg(k/4)}.
\]
Since $c = \lg( \sqrt{2} / \ln(4)) = \lg( 4 \sinh(\ln(2)/2) / \ln(4) )$, this becomes 
\[
  \le \ln(4)/k + 0.044 / (k \lg(k/4)).
\]

\paragraph{Overall discrepancy:} 
In total, we can bound the discrepancy $q := \Perf(\AlgJakub)$ of our algorithm (now including the factor of $k+1$) by
\[
  q \le (k+1) \max\left\{ \ln(\alpha^2)/k + O(1/k^2), \ln(4)/k + 0.044/(k \lg(k/4)) \right\}.
\]
Using $(k+1)/k = 1 + O(1/k)$ and 
\[
  \ln(\alpha^2) = \ln(4) \left( 1 + \frac{c}{\lg(k/4)} \right) \le \ln(4) + \frac{0.040}{\lg(k/4)},
\]
this bound can be simplified to
\[
  q \le \max\{ \ln(4) + 0.040/\lg(k/4) + O(1/k), \ln(4) + 0.044/\lg(k/4) + O(1/k) \},
\]
which proves Theorem~\ref{thm:algjacub}.
}

\section{Upper Bounds via Combinatorial Optimization}\label{sec:linear_programming}
In this section we show how to find upper bounds on the optimal discrepancy~$q^*(k)$ for fixed $k$. We do so by constructing cyclic algorithms using exhaustive enumeration of all short patterns in the case of very small $k$ or randomized local search on the patterns for larger $k$, combined with linear programming to optimize the checkpoint positions. This yields good algorithms as summarized in Table~\ref{table:small_k_upper_bounds}. In the following we describe our algorithmic approach.

\LV{\paragraph{Finding Checkpoint Positions:}}
First we describe how to find a nearly optimal cyclic algorithm given a pattern~$P$ and a scaling factor $\gamma$, i.e., how to optimize the checkpoint positions. To do so, we construct a linear program that is feasible if a cyclic algorithm with discrepancy $\lambda$ and scaling factor $\gamma$ exists. We use three kinds of constraints: We fix the ordering of the checkpoints, enforce that the $i$-th active checkpoint after one period is a factor $\gamma$ larger than the $i$-th initial checkpoint, and upper bound the discrepancy of each interval during the period by $\lambda$. We then use binary search to optimize $\lambda$. 

\begin{lemma} For a fixed pattern $P$ of length $n$ and scaling factor $\gamma$, let $q^*=\inf_A \Perf(A)$ be the optimal discrepancy among algorithms $A$ using $P$ and $\gamma$. Then finding an algorithm with discrepancy at most $q^* + \epsilon$ reduces to solving $O(\log \epsilon^{-1})$ linear feasibility problems with $O(nk)$ inequalities and $k+n$ variables.
\end{lemma}
\LV{
	\begin{proof} For a fixed pattern and scaling factor, we can tune the discrepancy of the algorithm by cleverly choosing the time points when to remove an old checkpoint and place a new one. By solving a linear feasibility problem we can check whether a cyclic algorithm with scaling factor $\gamma$ and pattern $P$ exists that guarantees a discrepancy of at most $\lambda$. We can then optimize over $\lambda$ to find an approximately optimal algorithm.

	We construct a linear program with the $k+n$ time points $(t_1,\ldots, t_{k+n})$ as variables (where we can set $t_k = 1$ without loss of generality). It uses three kinds of constraints. The first kind is of the form
	\[
		t_i \leq t_{i+1},
	\]
	for all $i\in [1,k+n)$. These constraints are satisfied if the checkpoint positions have the correct ordering, i.e.\ checkpoints with larger index are placed at later times.

	The second kind of constraints enforces the scaling factor. Since the pattern is fixed, we can compute at all steps which checkpoints are active. For $i\in[1,k]$ and $j\in[0,n]$, let $\tau_i^j$ be the variable of the $i$-th active checkpoint in step $j$ and let $\tau_0^j$ be $0$ for all $j$. It is easy to see that the algorithm has a scaling factor of $\gamma$ if the $i$-th active checkpoint in the last step is larger by a factor of $\gamma$ than in the first step. We encode this as constraints of the form
	\[
		\tau_i^n = \gamma \tau_i^0.
	\]
	Lastly we encode an upper bound of $\lambda$ for the discrepancy. Since the discrepancy of a cyclic algorithm is given by
	\[
		\max_{k < i \leq k+n} (k+1) \bar \ell_{t_i} / t_i,
	\]
	and each $\bar\ell_{t_i}$ can be expressed by a maximum over $k$ terms, we can encode a discrepancy guarantee of $\lambda$ with $nk$ constraints of the form
	\[
		\tau_{i+1}^j - \tau_i^j \leq \lambda \tau_k^j/(k+1), 
	\]
	for all $i\in [0,k)$ and $j\in [0,n]$.

	A feasible solution of these constraints fixes the checkpoint positions and hence, together with the pattern $P$, provides an algorithm with discrepancy at most $\lambda$. Using a simple binary search over $\lambda\in [1, 2]$ we can find an approximately optimal algorithm for this value of $\gamma$ and the pattern $P$. ~
	\end{proof}
}

\LV{\paragraph{Finding Scaling Factors:}
Next we show how to find scaling factors $\gamma$ for which algorithms with good discrepancy exist. We first show an upper bound for $\gamma$.
}
\SV{To find good algorithms without a fixed $\gamma$, we need the following lemma which is proved in the full paper.}
\begin{lemma}\label{lem:gamma_upper_bound} A cyclic algorithm with $k$ checkpoints, discrepancy $\lambda<k$, and a period length of $n$ can have scaling factor at most \SV{$\gamma \leq (1-\lambda/(k+1))^{-n}$.}
\LV{\[
	\gamma \leq \left(\frac{1}{1-\lambda/(k+1)}\right)^n.
\]}
\end{lemma}
\LV{
	\begin{proof} 
	Consider any checkpointing algorithm $A = (t,d)$ with $k$ checkpoints and discrepancy $\lambda$. At any time $t_i$, $i \ge k$, the largest interval has length $\bar \ell_{t_i} \ge t_i - t_{i-1}$, as there is no checkpoint in the time interval $[t_{i-1},t_i]$. Hence, we have 
	\begin{align*}
	  (k+1) \frac{t_i - t_{i-1}}{t_i} \le \lambda.
	\end{align*}
	Rearranging, this yields
	\begin{align*}
	  t_i \le \frac{1}{1 - \lambda/(k+1)} t_{i-1}.
	\end{align*}
	Iterating this $n$ times, we get
	\begin{align*}
	  t_{k+n} \le \left( \frac{1}{1 - \lambda/(k+1)} \right)^n t_k.
	\end{align*}
	Hence, for any cyclic algorithm (with discrepancy $\lambda$, $k$ checkpoints, and a period length of $n$) we get the desired bound on the scaling factor $\gamma = t_{k+n}/t_k$. ~
	\end{proof}
}

\LV{Since algorithms with discrepancy 2 are known~\cite{AhlrothPS11}, we can restrict our attention to $\lambda \le 2$. Hence, f}\SV{F}or any given pattern length $n$, Lemma~\ref{lem:gamma_upper_bound} yields an upper bound on $\gamma$, while a trivial lower bound is given by $\gamma > 1$. 
Now, for any given pattern $P$ we optimize over $\gamma$ using a linear search with a small step size over the possible values for $\gamma$. For each tested $\gamma$, we optimize over the checkpoint positions using the linear programming approach described above. 

\LV{
	\paragraph{Finding Patterns:} 
	For small $k$ and $n$, we can exhaustively enumerate all~$k^{n}$ removal patterns of period length $n$. Some patterns can be discarded as they obviously cannot lead to a good algorithm or are equivalent to some other pattern: No pattern that never removes the first checkpoint can be cyclic. Furthermore, patterns are equivalent under cyclic shifts, so we can assume without loss of generality  that all patterns end with removing the first checkpoint. Lastly, it never makes sense to remove the currently last checkpoint. Hence, for $k$ checkpoints there are at most $(k-1)^{n-1}$ interesting patterns of length $n$. This finishes the description of our combinatorial optimization approach.
}

\paragraph{Results:}
We ran experiments that \SV{exhaustively }try \SV{all }patterns up to length $k$ for $k\in [3,7]$. For $k=8$ we stopped\LV{ the search} after examining \SV{all }patterns of length $7$. For larger $k$ we used a randomized local search to find good patterns. The upper bounds we found are summarized in Table~\ref{table:small_k_upper_bounds}\SV{. 

We cannot prove (near) optimality for these algorithms, because we do not know whether short patterns (or any finite patterns) are sufficient, and whether the discrepancy behaves smoothly with $\gamma$ and a linear search can find a nearly optimal scaling factor. However, we tried all patterns of length $2k$ for $k\in [3,4,5]$ and found no better algorithm. Moreover, decreasing the step size in the linear search for $\gamma$ only yielded small improvements, suggesting that $\lambda$ is continuous in $\gamma$. This suggests that the bounds in Table~\ref{table:small_k_upper_bounds} are indeed close to optimal.
}%
\LV{, and for $k\leq 8$ the removal patterns and time points when to place new checkpoints can be found in Fig.~\ref{fig:small_k_positions}. Note that for $k=3$ this procedure re-discovers the golden ratio algorithm of Sect.~\ref{sec:golden_ratio}.}

Note that we can combine the results presented in Table~\ref{table:small_k_upper_bounds} with \SV{experimental results for } the algorithm \AlgKarl\ (Theorem~\ref{thm:algkarl}\LV{ and Fig.~\ref{fig:karls_algorithm_performance}}) to read off a global upper bound of $q^*(k) \le 1.7$ for the optimal discrepancy for \emph{any} $k$.

\begin{table}
\centering
\small
\begin{tabular}{l@{|}*{18}{c@{\hspace{1mm}}}}
$k$ 	&  3 	& 4 	& 5 	& 6 	& 7 	& 8 	& 9 	& 10 	& 15 	& 20 	& 30 	& 50 	& 100\\\hline
Discr. & 1.529 & 1.541 & 1.472 & 1.498 & 1.499 & 1.499 & 1.488 & 1.492 & 1.466 & 1.457 & 1.466 & 1.481 & 1.484
\end{tabular}

\caption{Upper bounds for different $k$. For $k<8$ all patterns up to length $k$ were tried. For $k=8$ all patterns up to length $7$ were tried. For larger $k$, patterns were found via randomized local search.}
\label{table:small_k_upper_bounds}
\end{table}
\LV{
	\begin{figure}
	\centering
	\includegraphics[width=\linewidth]{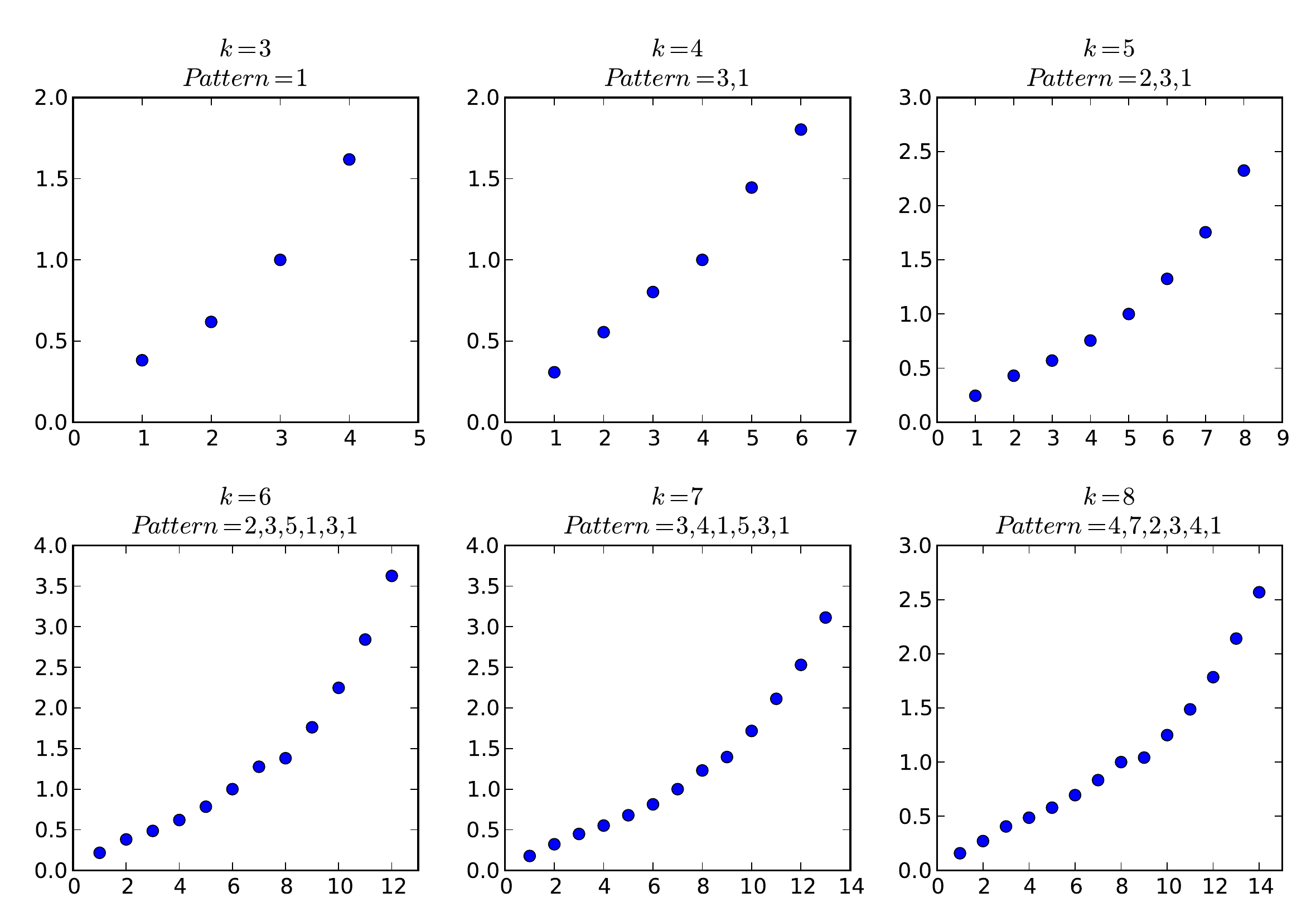}
	\caption{Time points where the $i$-th checkpoint is placed to achieve the bounds of Table~\ref{table:small_k_upper_bounds}. Time is on the $y$-Axis, iteration is on the $x$-Axis.}
	\label{fig:small_k_positions}
	\end{figure}
}
For a fixed pattern the method is efficient enough to find good checkpoint positions for much larger $k$. For $k\leq 1000$ we experimentally compared the algorithm \AlgKarl\ of Sect.~\ref{sec:karls_algorithm} with algorithms found for its pattern $(1,\ldots,k-1)$. The experiments show that for $k=1000$ \AlgKarl\ is within 4.5\% of the optimized bounds. For the algorithm \AlgJakub\ of Sect.~\ref{sec:jakubs_algorithm}, this comparison is even more favorable. For $k=1024$ the algorithm places its checkpoints so well that the optimization procedure improves discrepancy only by 1.9\%. \LV{The results are summarized in Fig.~\ref{fig:karls_algorithm_performance} and Fig.~\ref{fig:jakubs_algorithm_performance}.}

\LV{
	\begin{figure}
	\centering
	\includegraphics[width=\linewidth]{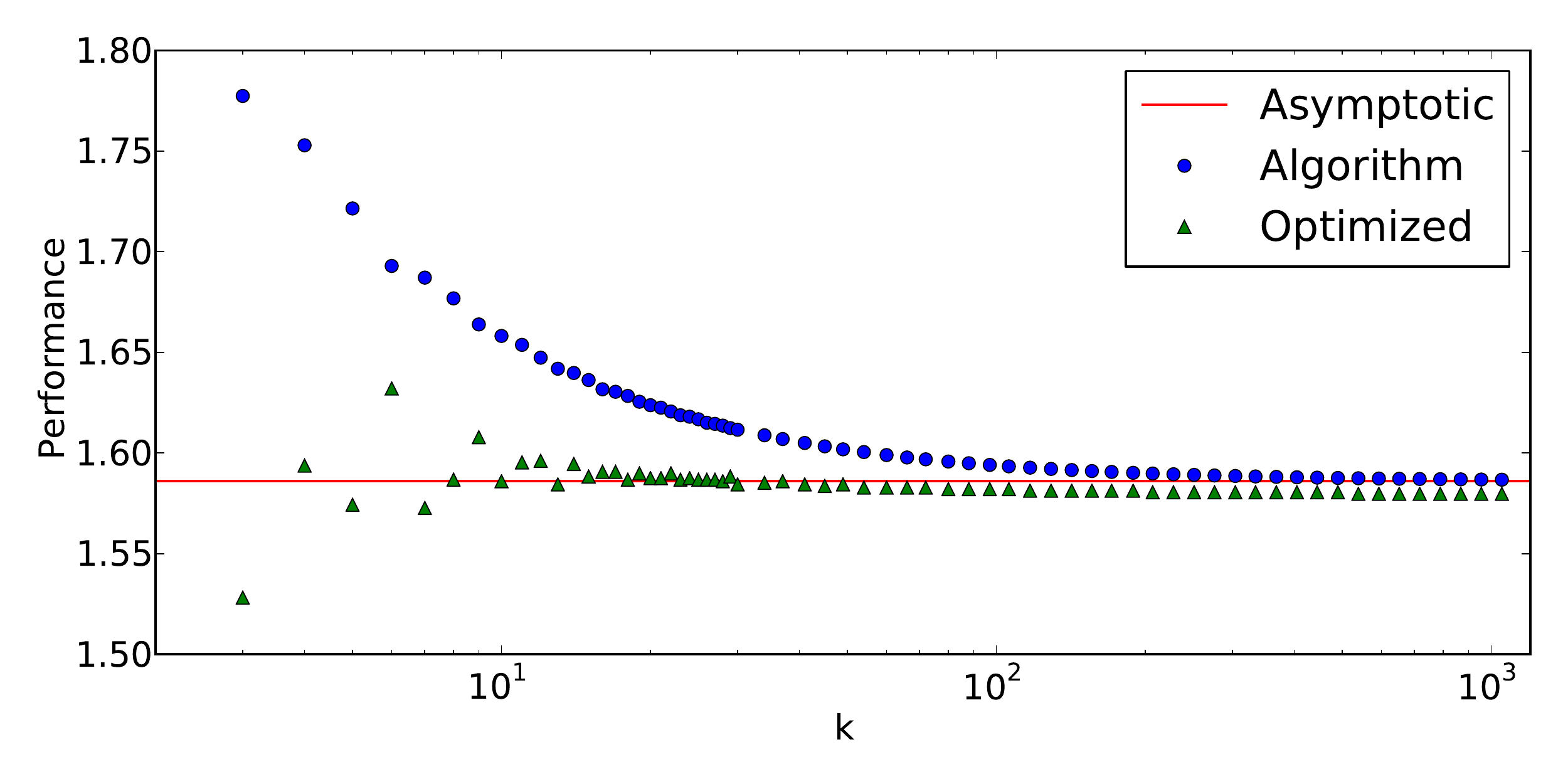}
	\caption{The discrepancy of algorithm \AlgKarl\ from Sect.~\ref{sec:karls_algorithm} for different values of $k$ compared with the upper bounds for its pattern found via the combinatorial method from Sect.~\ref{sec:linear_programming}. For large $k$ \AlgKarl\ is about 4.5\% worse.
	}
	\label{fig:karls_algorithm_performance}
	\end{figure}
}

\LV{
	\begin{figure}
	\centering
	\includegraphics[width=\linewidth]{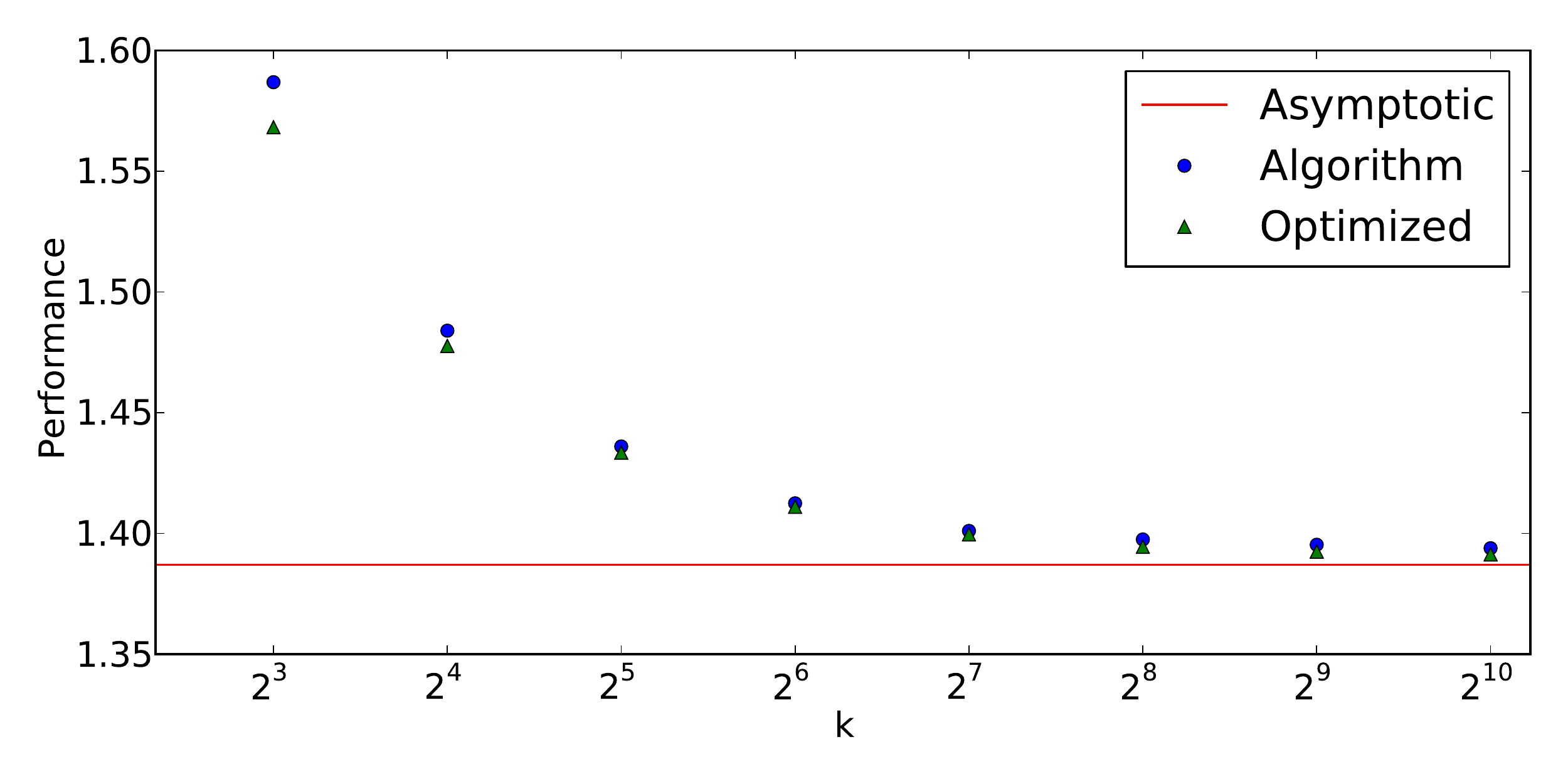}
	\caption{The discrepancy of the algorithm from Sect.~\ref{sec:jakubs_algorithm} for some values of $k$, compared with the upper bounds for its pattern found via the combinatorial method from Sect.~\ref{sec:linear_programming}. For $k=1024$, the optimization procedure finds a checkpoint placement with only 1.9\% better discrepancy.}
	\label{fig:jakubs_algorithm_performance}
	\end{figure}
}
\LV{
	\paragraph{Do we find optimal algorithms?}
	One could ask whether the algorithms from Table~\ref{table:small_k_upper_bounds} are optimal, or at least near optimal. There are two steps in above optimization algorithm that prevent this question to be answered positively. First, we are only optimizing over short patterns, and it might be that much larger pattern lengths are necessary for optimal checkpointing algorithms. Second, we do not know how smoothly the optimal discrepancy for fixed pattern $P$ and scaling factor~$\gamma$ behaves with varying $\gamma$, i.e., we do not know whether our linear search for $\gamma$ yields any approximation on the discrepancy $\lambda$. 
	However, in experiments we tried all patterns of length $2k$ for $k\in [3,4,5]$ and found no better algorithm than for the shorter patterns of length up to $k$. Moreover, smaller step sizes in the linear search for $\gamma$ lead only to small improvements, indicating that the discrepancy is continuous in $\gamma$. This suggests that the reported algorithms might be near optimal.
}

\section{Existence of Optimal Algorithms}\label{sec:existence}
In this section, we prove that optimal algorithms for the checkpointing problem exist, i.e., that there is an algorithm having discrepancy equal to the infimum discrepancy $q^*(k) := \inf_A \Perf(A)$ among all algorithms for $k$ checkpoints. 
\begin{theorem} \label{thm:existence}
  For each $k$ there exists a checkpointing algorithm $A$ for $k$ checkpoints with $\Perf(A) = q^*(k)$, i.e., there is an optimal checkpointing algorithm.
\end{theorem}

\LV{As we will see throughout this section, t}\SV{T}his a non-trivial statement. From the proof of this statement, we gain additional insight in the behavior of good algorithms. In particular, we show that we can assume without increasing discrepancy that for all $i$ the $i$-th checkpoint is set by a factor of at least $(1 + 1/k)^{\Theta(i)}$ later than the first checkpoint.

\LV{
	An initial set of checkpoints can be described by a vector $x = (x_1, \ldots, x_k)$, $0 \le x_1 \le \ldots \le x_k$. Since $x = (0, \ldots, 0)$ can never be extended to a checkpointing algorithm of finite discrepancy, we shall always assume $x \neq 0$. Denote by $X$ the set of all initial sets of checkpoints (described by vectors $x \neq 0$ as above), and by $X_0$ the set of all $x \in X$ with $x_k = 1$.

	We say that $A = (t,d)$ is an algorithm for an initial set $x \in X$ of checkpoints if $t_i = x_i$ for all $i \in [k]$. We denote by $q(x) := \inf_A \Perf(A)$, where $A$ runs over all algorithms for $x$, the \emph{discrepancy} of $x$. 
	An initial set $x \in X$ is called \emph{optimal} if $q(x) = \inf_{x \in X} q(x) = q^*(k)$. 
}

\LV{
	\begin{lemma}
	  Optimal initial sets of checkpoints exist.
	\end{lemma}

	\begin{proof}
	  Since the discrepancy of an initial set of checkpoints is invariant under scaling, that is, $q(x) = q(\lambda x)$ for all $x \in X$ and $\lambda > 0$, we have $\inf_{x \in X} q(x) = \inf_{x \in X_0} q(x)$. 
	  
	  It is not hard to see that $q(\cdot)$ is continuous on $X_0$: Let $x,x' \in X_0$ with $|x-x'|_\infty \le \varepsilon$ and consider an algorithm $A = (t,d)$ for $x$. We construct an algorithm $A' = (t',d)$ for $x'$ by setting $t_i' = t_i$ for $i > k$. Then $|\Perf(A) - \Perf(A')| \le 2 \varepsilon$, since any interval's length is changed by at most $2 \varepsilon$. This implies $|q(x) - q(x')| \le 2 \varepsilon$ and, thus, shows continuity of $q(\cdot )$. 
	  
	  Now, since $q(\cdot)$ is continuous on $X_0$ and $X_0$ is compact, there exists an $x \in X_0$ such that $q(x) = \inf_{x \in X_0} q(x) = q^*(k)$. ~
	\end{proof}

	An easy observation is that if some checkpointing algorithm leads to a vector $x$ of checkpoints at some time, then we may continue from there using any other algorithm for $x$. The discrepancy of this combined algorithm is at most the maximum of the two discrepancies.

	\begin{lemma}\label{ladd}
	  Let $A = (t,d)$ be a checkpointing algorithm. Let $i > k$. We call $q_{A,i} = \max_{j \in [k..i]} \bar \ell_{t_j} (k+1) / t_j$ the \emph{partial discrepancy} of $A$ observed in the time up to $t_i$. Assume that when running $A$, at time $t_i$ the checkpoints $x = (x_1, \ldots, x_k = t_i)$ are active. Let $A' = (t',d')$ be an algorithm for $x$. Then the checkpointing algorithm obtained from running $A$ until time $t_i$ and then continuing with algorithm~$A'$ is a checkpointing algorithm that has discrepancy at most $\max\{q_{A,i}, \Perf(A')\}$. If we run this combined algorithm only until some time $t'_j$, then the partial discrepancy observed till then is $\max\{q_{A,i},q_{A',j}\}$.
	\end{lemma}

	\begin{proof} Trivial. ~ \end{proof}

	The above lemma implies that in the following, we may instead of looking at an arbitrary time simply assume that the algorithm just started, that is, that the current set of checkpoints is the initial one.

	The following lemma shows that we can, without loss of discrepancy, assume that an algorithm for the checkpointing problem does not set checkpoints too close together. While also of independent interest, among others because it shows how to keep additional costs for setting and removing checkpoints low, we shall need this statement in our proof that optimal checkpointing algorithms exist.

	\begin{lemma}\label{lprogress}
	Let $A = (t,d)$ be an algorithm for the checkpointing problem with $\Perf(A) < k+1$. Then there is an algorithm $A' = (t',d')$ with the same starting position such that (i) $\Perf(A') \le \Perf(A)$ and 
	\begin{align*}
	  (ii) \; t'_{k+3} \ge t'_{k} \bigg(1 + \frac{\Perf(A)}{k+1-\Perf(A)} \bigg) \ge t'_k \left(1 + \frac 1k \right).
	\end{align*}
	\end{lemma}
	\begin{proof}
	Let $r = \Perf(A)/(k+1-\Perf(A))$ for convenience.
	By way of contradiction, assume that the lemma is false. Let $A$ be a counter-example such that $i := \min\{i \in \N \mid t_{k+i} \ge 1 + r\}$ is minimal 
	(the minimum is well-defined, since for any algorithm the sequence $(t_i)_i$ tends to infinity). Note that $i \ge 4$, since $A$ is a counter-example.

	Assume that there is a $j \in [1..i-1]$ such that $t_{k+j}$ in the further run of $A$ is removed (and replaced by the then current time $t_x$) earlier than both $t_{k+j-1}$ and $t_{k+j+1}$. Consider the Algorithm $A'$ that arises from $A$ by the following modifications. Let $t_y$ be the checkpoint that was removed to install the checkpoint $t_j$. Let $A'$ be the checkpointing algorithm that proceeds as $A$ except that $t_y$ is not replaced by $t_{k+j}$, but by $t_x$, and $t_{k+j}$ is never created. The only interval which could cause this algorithm to have a worse discrepancy than $A$ is $[t_{k+j-1},t_{k+j+1}]$. However, this interval contributes  $(k+1)(t_{k+j+1}-t_{k+j-1})/t_{k+j+1} \le (k+1) r/(1+r) \le \Perf(A)$ to the discrepancy of $A'$. Hence, $\Perf(A') \le \Perf(A)$ and $A'$ has fewer checkpoints in the interval $[1,1+r]$ contradicting the minimality of $A$. Thus, there is no $j \in [1..i-1]$ such that $t_{k+j}$ is removed earlier than both $t_{k+j-1}$and $t_{k+j+1}$ (*).

	We consider now separately the two cases that $t_{k+1}$ is removed earlier than $t_{k+i-2}$ and vice versa. Note first that $k+1 < k+i-2$ by assumption that $i \ge 4$. 

	Assume first that $t_{k+1}$ is removed (at some time $t_x$) earlier than $t_{t+i-2}$. Then $t_k$ must have been removed even earlier (at some time $t_y$), otherwise we found a contradiction to (*). Let $A'$ be an algorithm working identically as $A$, except that at time $t_y$ the checkpoint $t_{k+1}$ is removed (instead of $t_k$) and at time $t_x$ the checkpoint $t_k$ is removed (instead of $t_{k+1}$). Since the checkpoint at $t_{t+i-2}$ is still present, the only interval affected by this exchange, namely the one with $t_k$ as left endpoint, has length at most $r$. Hence as above, this contributes at most $\Perf(A)$ to the discrepancy of $A'$. The algorithm $A'$ has the property that there is a checkpoint in between $t_k$ and $t_{k+i-2}$ which is removed before these two points. The earliest such checkpoint, call it $t_{k+j}$, has the property that $t_{k+j}$ is removed earlier than both $t_{k+j-1}$ and $t_{k+j+1}$, contradicting earlier arguments. 

	A symmetric argument shows that also $t_{k+i-2}$ being removed before $t_{k+1}$ leads to a contradiction. Consequently, our initial assumption that $i \ge 4$ cannot hold, proving the claim. ~
	\end{proof}

	The following is a global variant of Lemma~\ref{lprogress}. It shows that any reasonable checkpointing algorithm does not store new checkpoints too often. 
}

\begin{theorem}
  Let $A = (t,d)$ be a checkpointing algorithm with $\Perf(A) < k-1$. Then there is an algorithm $A' = (t',d')$ with the same starting position such that (i) $\Perf(A') \le \Perf(A)$ and (ii) $t'_{i+3} \ge (1+1/k) \cdot t'_i$ for all $i \ge k$.
\end{theorem}
\LV{
	\begin{proof}
	  Let $j\ge k$ be the smallest index with a small jump, $t_{j+3} < (1+1/k) t_j$. Using Lemma~\ref{lprogress} (on the remainder of algorithm $A$ starting at time $t_j$) we can remove this small jump and get an algorithm $A'=(t',d')$ with $\Perf(A') \le \Perf(A)$ and $t'_{i+3} \ge (1+1/k) \cdot t'_i$ for all $k \le i \le j$, i.e., we patched the earliest small jump. Iterating this patching procedure infinitely often yields the desired algorithm. ~
	\end{proof}

	\begin{lemma}\label{lstep}
	For any optimal initial set $x = (x_1, \ldots, x_k)$, there is an algorithm $A = (t,d)$ such that (i) $q_{A,k+3} = \max_{j \in [k..k+3]} \ell_{t_j} (k+1) / t_j \le q^*(k)$, (ii) $t_{k+3} \ge t_k (1 + 1/k)$, and the set of checkpoints active at time $t_{k+3}$ is again optimal.
	\end{lemma}

	\begin{proof}
	By the definition of optimality, for each $n \in \N$ there is an algorithm $A^{(n)}$ for $x$ that has discrepancy at most $q^*(k) + 1/n$. 
	Let $(t_{k+1}^{(n)}, t_{k+2}^{(n)}, t_{k+3}^{(n)})$ denote the 
	corresponding next three checkpoints. By Lemma~\ref{lprogress}, we may assume that $t_{k+3}^{(n)} \ge t_k (1+1/k)$ for all $n \in N$.

	Note that (using the same arguments as in Lemma~\ref{lem:gamma_upper_bound}) 
	any algorithm having discrepancy at most $2.5$ satisfies $t_{k+i} \le 6^i t_k$ for any $k \ge 2$. Hence, $(t_{k+1}^{(n)}, t_{k+2}^{(n)}, t_{k+3}^{(n)})_{n \in \N_{\ge 2}}$ is a sequence in the compact space 
	$[t_k,6^3 t_k]^3$. This sequence has a convergent subsequence with limit $(t_{k+1}, t_{k+2}, t_{k+3})$. Also, since there are only finitely many values possible for $(d_{k+1}^{(n)}, d_{k+2}^{(n)}, d_{k+3}^{(n)})$, this subsequence can be chosen such that this $d$-tuple is constant, say $(d_{k+1}, d_{k+2}, d_{k+3})$. For this subsequence, also all $k+1$ intervals existing at the three times of interest converge. Consequently, the discrepancy caused by each of them also converges to a value upper bounded by $q^*(k)$. This defines the three steps of algorithm $A$, satisfying $q_{A,k+3} \le q^*(k)$. 

	Similarly, we observe that the set of checkpoints $x^{(n)}$ active at time $t^{(n)}_{k+3}$ when running algorithm $A^{(n)}$ has discrepancy at most $q^*(k) + 1/n$. Consequently, the active checkpoints we get from the limit checkpoints $(t_{k+1}, t_{k+2}, t_{k+3})$ and deletions $(d_{k+1}, d_{k+2}, d_{k+3})$ are again optimal.

	Finally, since all $t_{k+3}^{(n)} \ge t_k (1 + 1/k)$, this also holds for $t_{k+3}$.  ~
	\end{proof}

	We are now in position to prove the main result of this section, Theorem~\ref{thm:existence}.
	For this, we repeatedly apply Lemma~\ref{lstep}: We start with an optimal set of checkpoints~$x$. Then we run the algorithm delivered by Lemma~\ref{lstep} for three steps. This creates no partial discrepancy larger than $q^*(k)$ and we end up with another optimal set of checkpoints. From this, we continue to apply Lemma~\ref{lstep} and execute three steps of the algorithm obtained. By Lemma~\ref{ladd}, the partial discrepancy of the combined algorithm is again at most $q^*(k)$.
	Iterating infinitely, this yields an optimal algorithm, which proves Theorem~\ref{thm:existence}.
 }

\section{Lower Bound} \label{sec:lower_bound}

In this section, we prove a lower bound on the discrepancy of all checkpointing algorithms. For large $k$ we get a lower bound of roughly $1.3$, so we have a lower bound that is asymptotically larger than the trivial bound of 1. Moreover, it shows that algorithm \AlgJakub\ from Sect.~\ref{sec:jakubs_algorithm} is nearly optimal, as for large $k$ the presented lower bound is within $6\%$ of the discrepancy of \AlgJakub.

\begin{theorem} \label{thm:lower}
  All checkpointing algorithms with $k$ checkpoints have a discrepancy of at least \SV{$2 - \ln 2 - O(k^{-1}) \ge 1.306 - O(k^{-1})$.}
  \LV{\begin{align*}
    2 - \ln 2 - O(k^{-1}) \ge 1.306 - O(k^{-1}).
  \end{align*}}
\end{theorem}
\SV{
We present a sketch of the proof of the above theorem in the remainder of this section.

}
\LV{
  The remainder of this section is devoted to the proof of the above theorem.
}
  Let $A=(t,d)$ be an arbitrary checkpointing algorithm and let $q' := \Perf(A)$ be its discrepancy. For convenience, we define $q=kq'/(k+1)$ and bound $q$. Since $q<q'$ this suffices to show a lower bound for the discrepancy of $A$. For technical reasons we add a \emph{gratis checkpoint} at time $t_k$ that must not be removed by $A$. That is, even after the removal of the original checkpoint at $t_k$, there still is the gratis checkpoint active at $t_k$. Clearly, this can only improve the discrepancy. We analyze\LV{ the discrepancy of} $A$ from time $t_k$ until it deleted $k/(2q)$ of the initial checkpoints\footnote{To be precise we should round $\frac k{2q}$ to one of its nearest integers. When doing so, all calculations in the remainder of this section go through as they are; this only slightly increases the hidden constant in the error term $O(k^{-1})$.}. More formally, we let $t'$ be the minimal time at which the number of active checkpoints of $A$ contained in $[0,t_k]$ is $k - k/(2q)$. Note that we might have $t' = \infty$, if the checkpointing algorithm $A$ never deletes $k/(2q)$ points from $[0,t_k]$. However, in this case its discrepancy is lower bounded by $1.5$\LV{.}\SV{, since, for any $i>k$, there are at most $k-k/(2q)$ checkpoints available in the interval $(t_k,t_i]$.}
  
  \begin{lemma}
    If $t' = \infty$, then $\Perf(A) \ge 1.5$.
  \end{lemma} 
  \LV{
  \begin{proof}
    Consider a large $i > k$ and the algorithm's discrepancy at time $t_i$. By assumption, there are at most $k - k/(2q)$ active checkpoints in $(t_k,t_i]$. Hence, by comparing with an equidistant spread we can bound the discrepancy (at time $t_i$) by
    \begin{align*}
      \Perf(A) \ge \frac{k+1}{t_i} \cdot \frac{t_i - t_k}{k(1-1/(2q))} \ge \frac{2q}{2q-1} \Big(1 - \frac{t_k}{t_i} \Big).
    \end{align*}
    Letting $i \to \infty$, so that $t_i \to \infty$, we obtain
    \begin{align*}
      \Perf(A) \ge \frac{2q}{2q-1} \ge \frac{2 \Perf(A)}{2 \Perf(A) - 1},
    \end{align*}
    (by definition of $q$ and $x \mapsto \frac{2x}{2x-1}$ being monotonically decreasing). This inequality solves to the desired $\Perf(A) \ge 1.5$. ~
  \end{proof}
  }

  Hence, in the following we can assume that $t' < \infty$. 
  We partition the intervals that exist at time $t'$ into three types:
  \begin{enumerate}
    \item Intervals existing both at time $t_k$ and $t'$. These\LV{ intervals} are contained in $[0,t_k]$.
    \item Intervals that are contained in $[0,t_k]$, but did not exist at time $t_k$. These\LV{ intervals} were created by the removal of some checkpoint in $[0,t_k]$ after time~$t_k$.
    \item Intervals contained in $[t_k,t']$.
  \end{enumerate}
  Note that we need the gratis checkpoint at $t_k$ in order for these definitions to make sense, as otherwise there could be an interval overlapping $t_k$.

  Let $\mathcal{L}_i$ denote the set of intervals of type $i$ for $i \in\{1,2,3\}$, and set $k_i := |\mathcal{L}_i|$. Let $\mathcal{L}_2 = \{I_1,\ldots,I_{k_2}\}$, where the intervals are ordered by their creation times $\tau_1 \le \ldots \le \tau_{k_2}$. 
  Since each interval in $\mathcal{L}_2$ contains at least one deleted point we have \SV{$k_2\leq k/(2q)$,}
  \LV{\[
    k_2 \le \frac k{2q},
  \]}%
  and we set $m := \frac k{2q} - k_2$. Then $m$ counts the number of deleted checkpoints in $[0,t_k]$ that did not create an interval in $\mathcal{L}_2$, but some strict sub-interval of an interval in $\mathcal{L}_2$. \LV{We call these $m$ removed checkpoints \emph{free}.}
  
  \LV{We first bound the length of the intervals in $\mathcal{L}_1$ and $\mathcal{L}_2$.} \SV{Since the intervals in $\mathcal{L}_1$ exist at time $t_k$, we can bound their length as follows.}

  \begin{lemma}\label{lem:type1_bound} The length of any interval in $\mathcal{L}_1$ is at most $qt_k/k$.
  \end{lemma}
  \LV{
  \begin{proof}
  As all intervals in $\mathcal{L}_1$ already are present at time $t_k$ and the algorithm has discrepancy $q'$, we have for any $I \in \mathcal{L}_1$
  \[(k+1)|I| / t_k \le q' = (k+1)q/k.\]
  The bound follows. ~
  \end{proof}
  }

  \SV{The creation time $\tau_i$ of the $i$-th interval in $\mathcal{L}_2$ cannot be too late, as otherwise there would not be sufficiently many checkpoints in $[t_k,\tau_i]$ to guarantee discrepancy $q$. That allows to bound the length of these intervals.}

  \begin{lemma}\label{lem:type2_bound} The length of any interval $I_i\in \mathcal{L}_2$ is at most
  \[|I_i| \leq \frac{t_k}{k/q - m -i}.\]
  \end{lemma}
\LV{
  \begin{proof} As the algorithm has discrepancy $q'$, we know 
  \begin{equation}
  |I_i| \leq q\tau_i/k. \label{eq:type2_qual}
  \end{equation}
  In the following we bound $\tau_i$, the time of creation of $I_i$.
  At time $\tau_i$ there are at most $m+i$ intervals in $\mathcal{L}_3$, since at most $m$ free checkpoints and $i$ checkpoints from the creation of $I_1,\ldots,I_i$ are available. Comparing with an equidistant spread of $m+i$ checkpoints in $[t_k,\tau_i]$ and the algorithm's discrepancy, the longest interval $L$ in $[t_k,\tau_i]$ (at time $\tau_i$) has length 
  \[\frac{\tau_i-t_k}{m+i} \leq |L| \leq \frac{q\tau_i}{k}.\]
  Rearranging the outer inequality yields a bound on $\tau_i$ of 
  \begin{equation*}
  \tau_i \leq \frac{kt_k}{k-(m+i)q}.
  \end{equation*}
  Substituting this into \eqref{eq:type2_qual} yields the desired result. ~
  \end{proof}
}
  Furthermore, we need a relation between $k_1,k,m$, and $q$.

  \begin{lemma} \label{lem:relation_kmp}
    We have \SV{$k_1 = k + m - k/q + 1$.}
  \LV{\[
      k_1 = k + m - k/q + 1.
  \]}
  \end{lemma}
\LV{
  \begin{proof}
    As the intervals in $\mathcal{L}_1$ and $\mathcal{L}_2$ partition $[0,t_k]$, there are $k_1+k_2$ intervals left in $[0,t_k]$ at time $t'$. Note that each but one such interval has its left endpoint among the $k$ active checkpoints from time $t_k$ (the one exception having as left endpoint~0). Hence, there are $k_1+k_2 - 1$ checkpoints left in $[0,t_k]$. Comparing with the number $k_2+m$ of deleted checkpoints in $[0,t_k]$ until time $t'$ and their overall number $k$ yields
    \[
      (k_2 + m) + (k_1 + k_2 - 1) = k.
    \]
    Rearranging this and plugging in $k_2 = \frac k{2q} - m$ (which holds by definition of $m$) yields the desired result. ~
  \end{proof}
}
  Now we use our bounds on the length of intervals from $\mathcal{L}_1$ and $\mathcal{L}_2$ to find a bound on $q$. 
  Note that the intervals in $\mathcal{L}_1$ and $\mathcal{L}_2$ partition $[0,t_k]$, so that
  \[
  	t_k = \sum_{I \in \mathcal{L}_1} |I| + \sum_{I' \in \mathcal{L}_2} |I'|.
  \]
  Using Lemmas~\ref{lem:type1_bound} and~\ref{lem:type2_bound}, we obtain
  \[
  	t_k \leq k_1\frac{qt_k}{k} + \sum_{i = 1}^{k_2} \frac{t_k}{k/q - m - i}.
  \]
\SV{
  Plugging in Lemma~\ref{lem:relation_kmp} and simplifying yields
  \[
  q \geq 2 - m\frac pk - O(k^{-1}) - H_{k/q - m - 1} + H_{k/(2q)-1},
  \]
  where $H_n$ is the $n$th harmonic number. Now, observing $m \frac qk + H_{k/q - m - 1} \le H_{k/q - 1}$ and using the asymptotics of $H_n$, we get the desired bound $q \ge 2 - \ln(2) - O(k^{-1})$.
}
\LV{
  Substituting $k_1$ using Lemma~\ref{lem:relation_kmp} yields
  \begin{align}
  	t_k &\leq \left[k + m - k/q + 1\right] qt_k/k + \sum_{i=1}^{k/(2q)-m} \frac{t_k}{k/q - m - i}\notag\\
  		&= t_k \left(q - 1 + m\frac{q}{k} + O(k^{-1}) + \sum_{i = 1}^{k/(2q)-m} \frac{1}{k/q - m - i}\right). \label{eq:bound_involving_p}
  \end{align}
  Recall that $H_n=\sum_{1\leq i \leq n} i^{-1}$ is the $n$-th harmonic number. Rearranging \eqref{eq:bound_involving_p} yields
  \begin{align*}
  q &\geq 2 - m\frac qk - O(k^{-1}) - H_{k/q - m - 1} + H_{k/(2q)-1}.
  \end{align*}
  Observe that we have $m \frac qk + H_{k/q - m - 1} \le H_{k/q - 1}$, implying
  \begin{align*}
    q &\ge 2 + H_{k/(2q) - 1} - H_{k/q - 1} - O(k^{-1})  \\
    &\ge 2 + H_{k/(2q)} - H_{k/q} - O(k^{-1}),
  \end{align*}
  since we can hide the last summands of $H_{k/(2q)}$ and $H_{k/q}$ by $O(k^{-1})$. 
  In combination with the asymptotic behavior of $H_n = \ln n + \gamma + O(n^{-1})$, where $\gamma$ is the Euler-Mascheroni constant, we obtain 
  \begin{align*}
    q &\ge 2 + \ln(k/(2q)) - \ln(k/q) - O(k^{-1})  \\
    &= 2 - \ln(2) - O(k^{-1}).
  \end{align*}
  This finishes the proof of Theorem~\ref{thm:lower}. 
}

\bibliographystyle{plain}
\bibliography{bibliography}
\end{document}